\documentclass [12pt,letterpaper]{article}

\usepackage{pgfplots}
\usepackage{amssymb}
\usepackage{amsfonts}
\usepackage{multirow}
\usepackage{amsmath} 
\usepackage{amsthm}
\usepackage{wrapfig}
\usepackage{enumerate}
\usepackage{comment}
\usepackage{paralist}
\usepackage{natbib}
\usepackage{soul}
\usepackage[toc,page]{appendix}
\usepackage[letterpaper,left=25mm,right=25mm,top=30mm,bottom=25mm]{geometry}

\newcommand{\myF}[1]{ F(#1 | n-1) }
\newcommand{\myPsi}[1]{ \psi(n-1,#1,\varepsilon) }
\newcommand{\myV}[1]{ V(#1 | n-1) }
\newcommand{\myFs}[1]{ F^{\varepsilon}(#1 | n-1) }

\newcommand{\myVs}[1]{ V^{\varepsilon}(#1| n-1) }

\newtheorem{prop}{Proposition}
\newtheorem{lemma}{Lemma}
\newtheorem{cor}{Corollary}

\makeatletter
\renewcommand\fnum@figure[1]{\textbf{\figurename~\thefigure.~}}
\renewcommand\fnum@table[1]{\textbf{\tablename~\thetable.~}}
\makeatother

\begin{document}

\title{Behavioral Mistakes Support Cooperation \\ in an N-Person Repeated Public Goods Game}
\author{Jung-Kyoo Choi$^1$\thanks{jkchoi@knu.ac.kr} and  Jun Sok Huhh$^2$\\\\
\normalsize
$^1$ School of Economics and Trade, Kyungpook National University, Daegu 41566,. Korea \\
\normalsize
$^2$ NCSoft, Gyeonggi-do, 13494, Korea \\
}
\date{}
\maketitle

\begin{abstract}
This study investigates the effect of behavioral mistakes on the evolutionary stability of the cooperative equilibrium in a repeated public goods game. Many studies show that behavioral mistakes have detrimental effects on cooperation because they reduce the expected length of mutual cooperation by triggering the conditional retaliation of the cooperators. However, this study shows that behavioral mistakes could have positive effects. Conditional cooperative strategies are either neutrally stable or are unstable in a mistake-free environment, but we show that behavioral mistakes can make \textit{all} of the conditional cooperative strategies evolutionarily stable. We show that behavioral mistakes stabilize the cooperative equilibrium based on the most intolerant cooperative strategy by eliminating the behavioral indistinguishability between conditional cooperators in the cooperative equilibrium. We also show that mistakes make the tolerant conditional cooperative strategies evolutionarily stable by preventing the defectors from accumulating the free-rider's advantages. Lastly, we show that the behavioral mistakes could serve as a criterion for the equilibrium selection among cooperative equilibria. 
 
\vspace{0.25in} \noindent \textbf{Keywords}: Evolutionary stability of
cooperative strategy; Repeated public goods game; Behavioral
mistakes; Equilibrium selection\\
\vspace{0.1in}\noindent\textbf{JEL Classification} to be added
\end{abstract}

\section{Introduction}

Previous research has shown that in the context of a social dilemma, repeated interactions induce agents to consider the possibility of retaliation and thus to support a cooperative equilibrium \citep{Fudenberg_Maskin:1986,Fudenberg_Maskin:1990,Axelrod_Hamilton:1991,Taylor:1986}. This result is known as the folk theorem in economics or the reciprocity hypothesis in evolutionary biology. However, there are some criticisms of the repeated game approach. One of the issues is that repetition supports too many Nash equilibria, which raises the question of equilibrium selection \citep{Fudenberg_Maskin:1990,Binmore_Samuelson:1992,Axelrod:1997, Choi:2007,Sethi_Somanathan:2003}. In this regard, many efforts have been made to find sharper criteria for selecting efficient equilibria \citep{Binmore_Samuelson:1992,Fudenberg_Maskin:1990,Sethi_Somanathan:2003}.


Setting aside the problems of multiple equilibria and equilibrium selection, other problems remain, particularly from an evolutionary perspective. Previous studies show that cooperators are unlikely to evolve when they are rare in a population, especially when the group size is large \citep{Gintis:2006,Boyd_Richerson:1988}, and that even the efficient equilibria that are supported by game repetition do not satisfy dynamic stability \citep{Young_Foster:1991,Gintis:2006,Boyd_Lorberbaum:1987,Farrell_Ware:1989,YaoETAL:1996}. 

Among these two issues of inaccessibility and instability, the latter issue is the primary issue that we address in this paper. We will reconfirm the conclusion from previous studies that cooperation cannot be sustained in the long run, and we will investigate the role of behavioral mistakes in stabilizing the cooperative equilibrium. In this paper, we analyze a model of a repeated public goods game in an \textit{error-prone} environment where the individuals make unintended behavioral mistakes. The starting point of this paper is closely related to several previous studies. \citet{Joshi:1987} and \citet{Boyd_Richerson:1988} studied the effectiveness of repetition in a social dilemma situation where more than two agents are involved. Both studies assumed a repeated $n$-person public goods game structure in which conditional retaliation is triggered whenever the number of cooperators is less than the level of tolerance displayed by each individual. These studies introduced conditional cooperative strategies that differ in the number of defectors that are tolerated before retaliation is triggered. In this setting they showed that the only conditional cooperative strategy that is stable is the one that does not tolerate any defections, and the other conditional cooperative strategies that tolerate some defections are not stable at all. 

In this regard, many studies note that the conditional cooperative strategy that does not tolerate any defections is only neutrally stable in that it can remove defectors from its population when the game is repeated with sufficiently high probability, but it allows other cooperative strategies with higher tolerance levels to remain in its population. It is because, in the absence of defectors in the population, all of the conditional cooperative strategies, regardless of their tolerance levels, are behaviorally indistinguishable and, as a result, receive the same payoffs. Mutant cooperative strategies with higher tolerance levels enter the population and can cause a random drift as long as there are no defectors in the population. This process will ultimately result in a population state having too many tolerant conditional cooperators, which makes the population vulnerable to the invasion of defectors into the population. This is drift problem well known in evolutionary biology. In other words, a population that is composed of individuals using the conditional cooperative strategy with the least tolerance is subject to drift by which the dynamic path is ultimately led to a situation where the population is occupied only by defectors. Cooperation cannot be sustained in the long run and this dynamic instability is due to the fact that none of the cooperative strategies are evolutionarily stable \citep{Samuelson:2002,Choi:2007}. 

Here, we introduce a crucial question that remains unaddressed. Are the previous conclusions modified if the individual agents are allowed to make mistakes? There have been some studies that examine the effect of mistakes on the cooperative equilibrium in a repeated game. Most of these studies, however, only partially address the effect of behavioral errors, primarily in the context of their detrimental effects on cooperation \citep{Boyd:1989,Bendor_Mookherjee:1987}. 

In this paper, we will reframe the model in an error-prone environment and show how introducing behavioral mistakes significantly modifies the conclusions that are drawn from an error-free assumption. As many studies have shown, mistakes have a detrimental effect on cooperation because they trigger retaliation and reduce the expected duration of the game. However, we will show that, in the presence of behavioral mistakes, the conditional cooperative strategy can become evolutionarily stable. In Section \ref{setting}, we set up a model of a public goods game and introduce the dynamic stability issue of cooperative equilibrium. We identify the dynamic property of cooperative equilibrium supported by repetition and show why neutral stability does not guarantee the sustainability of cooperation in the long run. In Section \ref{sec error}, we present a model of a public goods game where the cooperators make mistakes and show that the cooperators' mistakes make \textit{all} of the conditional cooperative strategies evolutionarily stable if the game is repeated with a sufficiently high probability. We also show that the rate of mistakes serves as an equilibrium selection criterion. Lastly, in Section \ref{sec general case} we show that the conclusions drawn in the previous sections still hold in the more general cases.

\section{An $n$-Person Public Goods Game without Errors\label{setting}}

Suppose that groups of $n$ individuals are randomly drawn from the population to interact repeatedly in an $n$-person public goods game\footnote{The model of this section is based on \citet{Boyd_Richerson:1988}.}. Each individual's payoff depends on his or her action and the actions of the $n-1$ other individuals in the group. Let $j$ be the number of cooperators among the other members in the group. Define $F(C|j)$ and $F(D|j)$ as the payoff of the one-shot game for cooperation and defection, respectively, when there are $j$ cooperators among the $n-1$ other individuals. If the player is a cooperator, then the number of cooperators in the group becomes $j + 1$; otherwise, the number of cooperators remains $j$. Therefore, we have  
\begin{align}
\begin{aligned}
F(C|j)&=\frac{b\cdot (j+1)}{n}-c\\
F(D|j)&=\frac{b\cdot j}{n}.
\end{aligned}
\end{align}
%

Assume that this public goods game is repeated in each group with probability $\delta$. For simplicity, assume that individuals do not discount the future. Furthermore, consider a situation with incomplete information: players know how many cooperators in their group were in the previous stage, but they do not know exactly who defected. 

Assume that the individuals only remember the outcome of the last stage. The individuals now decide whether to cooperate in this stage conditional on how many cooperators were in the previous stage. Define individual $i$'s pure strategy set, $S_i = \{T_k| 0 \leq k\leq n\}$ for an integer $k$. The $T_k$ strategy is similar to a tit-for-tat strategy in a two-person prisoner's dilemma and can be expressed as follows: cooperate if $k$ or more of the other $n-1$ individuals cooperate in the group during the previous stage and defect otherwise. We will call a player who uses the $T_k$ strategy a $k$-cooperator. Thus, for example, a person following the $T_{n-1}$ strategy (i.e., an $[n-1]$-cooperator) will cooperate only if every other individual cooperates, a person with the $T_0$ strategy will cooperate unconditionally, and a person with the $T_n$ strategy is an unconditional defector because it is impossible to have $n$ cooperators among $n-1$ the others. The subscript $k$ (if $k < n$), therefore, refers to a player's degree of willingness to retaliate. We call $T_{n-1}$ the hardest (or the least tolerant) strategy because the $[n-1]$-cooperator does not tolerate any defectors in their group. 

Consider a situation where every conditional cooperator has the same $k$, and a small number of $T_n$ players (unconditional defectors) appear in this population. Let the share of mutants in the post-entry population be $\mu$, where $\mu\in (0,1)$.

When the game is repeated with the probability $\delta$, the payoff for $T_k$ and $T_n$ from the repeated game is as follows: 
\begin{align}\label{payoff}
\begin{aligned}
V(T_k|j)&=
\begin{cases}
\frac{1}{1-\delta}F(C|j)&\text{if $j> k$}\\
\frac{1}{1-\delta}F(C|j)&\text{if $j= k$}\\
F(C|j)+\frac{\delta}{1-\delta}F(D|0)&\text{if $j<k$}
\end{cases}\\
V(T_n|j)&=
\begin{cases}
\frac{1}{1-\delta}F(D|j)&\text{if $j> k$}\\
F(D|j)+\frac{\delta}{1-\delta}F(D|0)&\text{if $j=
k$}\\
F(D|j)+\frac{\delta}{1-\delta}F(D|0)&\text{if $j<
k$}.
\end{cases}
\end{aligned}
\end{align}
After playing a repeated public goods game, the players have the chance to update their strategy. Let $p_k$ and $p_n$ be the population frequency of the $k$-cooperators and the unconditional defectors, respectively. We assume that the strategy-updating occurs in the following way: Whether the unconditional defectors can increase their frequency in a population where $p_k = 1-\mu$ and $p_n = \mu$ depends on whether the expected payoff for the $T_n$ strategy is greater than that for the $T_k$ strategy. When $p_k = 1-\mu$, the expected payoffs for the $T_k$ strategy and the $T_n$ strategy are  

\begin{align}
\begin{aligned}
W(T_k|p_k=1-\mu)&=\sum_{j=0}^{k-1}V(T_k|j)m(j, p_k)+V(T_k|k)m(k, p_k)+\sum_{j=k+1}^{n-1}V(T_k|j)m(j, p_k)\\
W(T_n|p_k=1-\mu)&=\sum_{j=0}^{k-1}V(T_n|j)m(j, p_k)+V(T_n|k)m(k,
p_k)+\sum_{j=k+1}^{n-1}V(T_n|j)m(j, p_k),
\end{aligned}
\end{align}
\normalsize
where $m(j,p_k)$ is the probability that an individual finds him/herself in a group in which there are $j$ other individuals following the $T_k$ strategy, which can be written as 
\begin{align*}
m(j, p_k)=\binom{n-1}{j}p_k^{j}(1-p_k)^{n-1-j}.
\end{align*}
In general, when a mutant strategy $T_{k'}$ appears with a share of $\mu$ in a population that is homogeneously composed of individuals following the $T_k$ strategy (where $k'\neq k$), the $T_k$ strategy is neutrally stable if and only if there exists a range of $\mu\in [0,\bar{\mu}]$ that satisfies $W(T_k|p_k = 1-\mu)\geq W(T_{k'}|p_k = 1-\mu)$ and that is evolutionarily stable when the condition is satisfied with the strict inequality. This stability condition is equivalent to $W(T_k|p_k = 1) \geq W(T_{k'}|p_k = 1)$, and $W(T_k|p_{k'} = 1) > W(T_{k'}|p_{k'} = 1)$ if $W(T_k|p_k = 1) = W(T_{k'}|p_k = 1)$  \citep{Weibull:1997}.  

\begin{prop}\label{neut t n-1}
In an n-person public goods game where the game is repeated with probability $\delta$, the $T_{n-1}$  strategy is neutrally stable.
\end{prop}
\begin{proof}
For the $T_{n-1}$  strategy to be neutrally stable, there should exist a $\bar{\mu}$  such that for $\mu\in[0,\bar{\mu}]$, 
\begin{align}
W(T_{n-1}|p_{n-1}=1-\mu)\geq W(T_{k'}|p_{n-1}=1-\mu),\quad \forall
k'\neq n-1
\end{align}
is satisfied. 

(i) For $k'<n-1$. Suppose that a mutant with $T_{k'}$  (where $k'<n-1$) appears in a population where $p_{n-1}=1$. In the absence of a $T_n$  strategy, all of the conditional cooperative strategies (i.e., all of the $T_{k'}$  strategies where $k'<n-1$ and the $T_{n-1}$  strategy) are behaviorally indistinguishable among themselves and receive the same payoffs. In other words, any $T_k$ strategy players (if $k<n$ begin the game by playing $C$ and continue playing $C$  until the game ends because there would be no defection in the group in the absence of $T_n$. Therefore, we have $W(T_{n-1}|p_{n-1}=1-\mu)=W(T_{k'}|p_{n-1}=1-\mu)=(b-c)/(1-\delta)$  for all $k'<n-1$ and all $\mu$.

(ii) For $k'=n$. Suppose that a mutant having $T_n$  appears in a population that is homogeneously composed of the $T_{n-1}$ strategy players.  Where $p_{n-1}=1$ we have $m(n-1, p_{n-1})=m(n-1,1)$. Therefore, the payoff to the $T_{n-1}$ is $W(T_{n-1}|p_{n-1}=1)=\frac{1}{1-\delta}F(C|n-1)=\frac{1}{1-\delta}(b-c)$, and the payoff to the $T_n$  is $W(T_n|p_{n-1}=1)=F(D|n-1)+\frac{\delta}{1-\delta}F(D|0)=\frac{b(n-1)}{n}$. Now we obtain the condition, $\frac{1}{1-\delta}(b-c)>\frac{b(n-1)}{n}$, which can be rearranged to $\delta>\frac{c-(b/n)}{b-(b/n)}$. Because $c<b$, there exists a $\delta<1$  that satisfies the condition. Therefore, when the game is repeated with a sufficiently high $\delta$, we have $W(T_{n-1}|p_{n-1}=1)>W(T_n|p_{n-1}=1)$. 
\end{proof}
Applying the same logic to the stability of the other conditional cooperative strategies, we can easily show that the other softer $T_k$ strategies (where $k<n-1$) are not stable. The $T_k$ strategies with $k<n-1$ tolerate the defector's free riding behavior as long as the number of defectors in the group does not exceed $n-k-1$. As a result, these strategies allow the defectors to invade the $T_k$ population.   
\begin{cor}\label{stability Tk}
In an $n$-person public goods game where the game is repeated with probability $\delta$, the $T_k$ strategies where $k<n-1$, are not stable. 
\end{cor}
\begin{proof}
Suppose that a mutant, $T_n$, appears in a population that is entirely composed of $T_k$ strategy players ($k<n-1$), i.e., $p_k=1$. The expected payoff for the $T_k$ strategy when $p_k$ is equal to 1 is $W(T_k|p_k=1)=(b-c)/(1-\delta)$ and the expected payoff to the mutant $T_n$ strategy when $p_k$ is equal to 1 is $W(T_k|p_k=1)=\frac{b(n-1)}{n}/(1-\delta)$. Because $\frac{b(n-1)}{n}>b-c$ by the assumption of the payoff structure, we have $W(T_{k}|p_{k}=1-\mu)< W(T_n|p_{k}=1-\mu)$ for a sufficiently small \(\mu\) and for all $k < n-1$. 
\end{proof}
That the $T_{n-1}$  strategy is only neutrally stable and not evolutionarily stable raises interesting issues regarding the dynamic property of the equilibrium. Consider a population homogeneously that is composed of individuals using $T_{n-1}$, the least tolerant cooperatrive strategy. Because the $T_{n-1}$  individuals do not allow any free riders in their group, the defection strategy cannot invade this population if the probability of game repetition is sufficiently high. The defectors will be retaliated against immediately and eliminated by the numerically predominant $T_{n-1}$  individuals. However, other $T_{k<n-1}$  strategies can invade the $T_{n-1}$  population because all of the conditional cooperators are behaviorally indistinguishable and receive the same payoffs when the defectors do not exist. And this indistinguishability allows the $T_{k<n-1}$  strategies to remain in the $T_{n-1}$  population. Once a sufficient number of $T_{k<n-1}$ strategies accumulate in the population, the defectors can invade and gain benefits from free riding because the $T_{k<n-1}$  strategies tolerate some defections. Therefore, due to this behavioral indistinguishability between $T_{n-1}$  and $T_{k<n-1}$  in the absence of a defection strategy, the equilibrium state where only the $T_{n-1}$  strategy is present will not persist over long periods.

In the following section, we introduce behavioral mistakes and show how these mistakes change the previous conclusions obtained from an error-free environment.

\section{An $n$-Person Public Goods Game and Behavioral Errors\label{sec error}}
\subsection{Introducing Mistakes}

Let us assume that only conditional cooperators make mistakes (to make a contribution, a cooperator needs to take an action, whereas a defector simply does nothing). For simplicity, a player who made a mistake does not know that it is he/she that made the mistake. Cooperators only know the total number of cooperations in their group, and they determine whether they would cooperate or defect in the next round depending on this number.  

Because some cooperators defect by mistake, the payoff to individuals now depends on the number of errors made by the cooperators as well as the number of cooperators in a group. Let $j$  be the number of cooperators among $n-1$  other members in a group. Let $\psi(j, q, \varepsilon)$  be the probability that $q$  individuals among the $j$  other cooperators in a group make mistakes when the probability of making a mistake is $\varepsilon$.
%
\begin{align*}
\psi(j, q,\varepsilon)=\binom{j}{q}\varepsilon^{q}(1-\varepsilon)^{j-q}~\text{for $\varepsilon > 0$}.
\end{align*}
The payoffs from the one-shot game, given that $j$  among $n-1$  other individuals choose cooperation, now become
\begin{align*}
F^\varepsilon (C|j)&=(1-\varepsilon)\left( \sum_{q=0}^{j}\psi(j, q,
\varepsilon)(b\times {j-q+1\over n}-c)\right)\nonumber+ \varepsilon
\left(\sum_{q=0}^{j}\psi(j, q,
\varepsilon)(b\times {j-q\over n}\right)\\
F^\varepsilon(D|j)&=\sum_{q=0}^{j}\psi(j, q,
\varepsilon)\left(b\times {j-q\over n}\right).
\end{align*}
\normalsize{}Suppose that the game is repeated with probability $\delta$. Let $V^\varepsilon(T_k|j)$ be the payoff for the $T_k$ strategy when $j$ among the $n-1$ other members in his or her group are the conditional cooperators and the conditional  cooperators play defection by mistake with probability $\varepsilon$. Then, the behavioral mistakes affect $V^\varepsilon(T_k|j)$ via whether the conditional cooperators maintain their cooperation in the following stages as well as via a one shot payoff, $F^\varepsilon(C|j)$ and $F^\varepsilon(D|j)$.

\subsection{\label{section t n-1 stability}The stability of the $T_{n-1}$ strategy}

We showed in Proposition \ref{neut t n-1} that the $T_{n-1}$ strategy is neutrally stable (not evolutionarily stable). Now we will show how behavioral mistakes affect the stability of the $T_{n-1}$ strategy. Consider a population that is homogeneously composed of individuals using the $T_{n-1}$ strategy.  A player with the $T_{n-1}$ strategy receives $\myFs{C}$ from the first round and maintains cooperation in the next round only if there is no defection in the current round. Therefore, the mutual cooperation continues only if there are neither defectors' defecting nor cooperators' mistakenly defecting.  In this environment, the payoff for the $T_{n-1}$ strategy from the repeated game is 
 
\begin{eqnarray}\label{Tn-1 Tn-1}
V^\varepsilon(T_{n-1}|n-1)=F^\varepsilon(C|n-1)+(1-\varepsilon)^n\delta V(T_{n-1}|n-1, \varepsilon).
\end{eqnarray}
\normalsize
Rearranging the above condition gives 
\begin{eqnarray}\label{v_cond}
V^\varepsilon(T_{n-1}|n-1)={F^\varepsilon(C|n-1)\over 1-(1-\varepsilon)^n\delta},
\end{eqnarray}
\normalsize{}where $(1-\varepsilon)^n$  is the probability that no conditional cooperators make mistakes, that is, the probability that the periods of mutual cooperation continue without triggering retaliation by any mistakes.
Suppose that there appears a mutant with the $T_n$ strategy in this population. This strategy will immediately trigger the $T_{n-1}$  players' retaliation. Therefore, the payoff for the $T_n$ strategy is
\begin{eqnarray} \label{eq:Tn Tn-1}
\myVs{T_n}=\myFs{D}+\delta\sum_{q=0}^{n-1}\psi(n-1, q, \varepsilon)\times 0.
\end{eqnarray}
\normalsize

Suppose that there appears a mutant with the $T_k$ strategy ($k<n-1$) in the $T_{n-1}$ population. This $k$-cooperator plays $C$ in the first round and keeps cooperating as long as the number of defections that are mistakenly played by the conditional cooperators is less than or equal to $n-k-1$. The payoff for the $T_k$ strategy from the repeated game in a population where all of the other players are $T_{n-1}$ becomes
%
\begin{eqnarray}\label{Tk Tn-1}
V^\varepsilon(T_k|n-1)=F^\varepsilon(C|n-1,
\varepsilon)&+&(1-\varepsilon)^n\delta V^\varepsilon(T_k|n-1)\nonumber\\
&+&\delta\sum_{q=1}^{n-k-1}\psi(n-1,q,
\varepsilon)F^\varepsilon(C|0), 
\end{eqnarray}
\normalsize
where $k<n-1$. The last term on the right-hand side expresses the payoff for the $T_k$  strategy in the next round when $T_{n-1}$  players have already withdrawn their cooperation and $T_{k<n-1}$  players continue cooperating. This situation occurs when the number of mistakes in one stage is greater than 1 and less than or equal to $n-k-1$, and $k$ cooperators will also withdraw cooperation thereafter. The rearrangement gives 
%
\begin{eqnarray}
V^\varepsilon(T_k|n-1)&=&{F^\varepsilon(C|n-1,
\varepsilon)+\sum_{q=1}^{n-k-1}\psi(n,q,\varepsilon)F^\varepsilon(C|0)\over 1-(1-\varepsilon)^n\delta}\nonumber\\
&=&{F^\varepsilon(C|n-1,
\varepsilon)+\sum_{q=1}^{n-k-1}\psi(n,q,\varepsilon)\delta(1-\varepsilon)({b\over
n}-c)\over 1-(1-\varepsilon)^n\delta},\nonumber\\ 
\end{eqnarray}
\normalsize
where $k<n-1$. Here, the behavioral mistakes have two effects on the evolution of the conditional cooperative strategies. First, the mistakes reduce the period of mutual cooperation and have a detrimental effect on the evolution of cooperation. Suppose that a group is entirely composed of $T_{n-1}$  players. Without mistakes, all of the members continue cooperating until the game ends, and the expected length of the period of mutual cooperation is $\frac{1}{1-\delta}$. When the players make mistakes with probability $\varepsilon$, even one mistake will trigger retaliation from the other $T_{n-1}$  players, and the mutual cooperation ends. Therefore, the period of mutual cooperation is reduced to $\frac{1}{1-(1-\varepsilon)^n\delta}$  and the value of $V(T_{n-1}|n-1, \varepsilon)$ in (\ref{v_cond}) becomes smaller.

Second, in an error-prone environment, each conditional cooperator reacts differently to the others' mistakes according to his/her tolerance level. For example, $T_{n-1}$  players immediately begin retaliation by withdrawing cooperation, $T_{n-2}$  players tolerate one mistake, and $T_{n-h}$  players tolerate $h-1$  mistakes. Note that the behavioral indistinguishability between the $T_k$  strategies (where $k<n$) disappears in the absence of $T_n$, and the payoffs for the $T_k$  strategies are no longer the same, even in a situation where the $T_n$  strategy does not exist.
As mentioned in the previous section, the dynamic instability occurs because the $T_{n-1}$  strategy is only neutrally stable, not evolutionarily stable. In other words, the dynamic problem occurs because both $T_{n-1}$  and $T_{k<n-1}$  receive the same payoff in the absence of a defection strategy in an error-free environment. As soon as the possibility of making errors is introduced, the problem of indistinguishability between all of the conditional cooperative strategies in the absence of a defect strategy can be solved because each conditional cooperative strategy reacts differently to the error that the other players make.

The following proposition shows that mistakes can eliminate the indistinguishability between the $T_{k<n}$  strategies so that the state where all of the individuals are $T_{n-1}$  players is evolutionarily stable.

\begin{prop} \label{prop t n-1}
In an n-person public goods game where conditional cooperators have a $T_k$  strategy \upshape{(}\itshape where $k \in \{0, \cdots, n-1 \}$\upshape{)}\itshape , the $T_{n-1}$  strategy is evolutionarily stable when the probability of game repetition is sufficiently close to 1 and the probability of making mistakes is positive but sufficiently small. 
\end{prop}
\begin{proof}
The condition for the $T_{n-1}$ strategy  being evolutionarily stable with respect to some other strategy is 
%
\begin{eqnarray}\label{evol_cond}
\begin{array}{ccc}
W(T_{n-1}| p_{n-1}=1)&>&W(T_{k<n-1}| p_{n-1}=1)\\
&\textrm{and}&\\
W(T_{n-1}|p_{n-1}=1)&>&W(T_n| p_{n-1}=1).
\end{array}
\end{eqnarray}
\normalsize
In the proof of Proposition \ref{neut t n-1}, we have already shown that the first condition in Eq (\ref{evol_cond}) is not satisfied with a strict inequality when the players do not make any mistakes. That is, when the probability of making mistakes is zero, we have $W(T_{n-1}|p_{n-1}=1)= W(T_{k<n-1}|p_{n-1}=1)$. However, if the probability of making mistakes is positive, then we have 
\begin{eqnarray}\label{W Tn-1 Tn-1}
W(T_{n-1}| p_{n-1}=1)=V^\varepsilon(T_{n-1}|n-1)m(n-1,1)=\frac{F^\varepsilon(C|n-1)}{1-(1-\varepsilon)^n\delta}~.
\end{eqnarray}
\normalsize{}

(i) Suppose that a mutant with the $T_{k'}$  (where $k'<n-1$ ) appears in a population where $p_{n-1}=1$. With the possibility of mistakes, the expected payoff to the $T_{k'}$ when $p_{n-1}=1$ is  $W(T_{k'}|p_{n-1}=1)=V^\varepsilon(T_{k'}|n-1)m(n-1,1)$. According to Eq (\ref{Tk Tn-1}), the above equation can be expressed as 
\begin{eqnarray}\label{W Tk Tn-1}
W(T_{k'}|p_{n-1}=1)&=&V^\varepsilon(T_{k'}|n-1)m(n-1,1)\nonumber\\
&=& \frac{F^\varepsilon(C|n-1)+\sum_{q=1}^{n-k'-1}\psi(n,q,\varepsilon)\delta(1-\varepsilon)(\frac{b}{n}-c)}{1-(1-\varepsilon)^n\delta}.
\end{eqnarray}
\normalsize{}Because $0<\frac{b}{n}<c$  by the assumption, comparing Eq (\ref{W Tn-1 Tn-1}) and Eq (\ref{W Tk Tn-1}) gives 
\begin{eqnarray}
W(T_{n-1}|p_{n-1}=1)>W(T_{k'}|p_{n-1}=1)
\qquad \textrm{for any }k'<n-1,
\end{eqnarray}
when the probability of making a mistake is positive.

(ii) Suppose that a mutant having the $T_n$ strategy appears in a population where $p_{n-1}=1$. When $\varepsilon>0$, the payoff to $T_n$  from the repeated public goods game in this situation is  
\begin{eqnarray}
W(T_n|p_{n-1}=1)&=&V^\varepsilon(T_n|n-1)m(n-1, 1)\nonumber\\
&=&F^\varepsilon(D|n-1).
\end{eqnarray}
\normalsize{}Therefore, for $W(T_{n-1}|p_{n-1}=1)> W(T_{n}|p_{n-1}=1)$, we should have  
\begin{eqnarray}
V^\varepsilon(T_{n-1}|n-1)={F^\varepsilon(C|n-1)\over 1-(1-\varepsilon)^n\delta}>F^\varepsilon(D|n-1)=V^\varepsilon(T_n|n-1).
\end{eqnarray}
\normalsize{}The above can be satisfied if 
\begin{eqnarray}
(1-\varepsilon)^n\delta>1-{F^\varepsilon(C|n-1, \varepsilon)\over
F^\varepsilon(D|n-1, \varepsilon)}.
\end{eqnarray}
\normalsize
We can easily show that $\frac{F^\varepsilon(C|n-1, \varepsilon)}{F^\varepsilon(D|n-1, \varepsilon)}=\frac{n}{n-1}(1-\frac{c}{b})$ and that this is always less than 1 by the assumption $\frac{b}{n}-c<0$. That is, the right-hand side of the condition is always less than 1. Therefore, with a sufficiently large $\delta$ and sufficiently a small $\varepsilon$, the above condition is satisfied.   
\end{proof}
As long as there is a probability of making mistakes, the indistinguishability between the hardest cooperators and the other softer cooperators disappears, which makes the $T_{n-1}$ equilibrium state evolutionarily stable under certain parameter values. The $T_k$  strategies ($k<n-1$) are now behaviorally distinguishable from the $T_{n-1}$ strategy  because the conditional cooperative strategies are, depending on the level of $k$, different in their responding to the mistakes, even though there are no defectors in the population. In other words, in an environment where the players make mistakes, the $T_k$  strategies cannot invade a population that is homogeneously composed of individuals using the $T_{n-1}$  strategy.

\subsection{\label{section t k stability}The stability of the $T_k$ strategy {\upshape(}$k<n-1${\upshape)}.}

Mistakes also affect the evolutionary stability of the softer strategies. Note that the $T_k$ strategies when $k<n-1$ are not stable at all in an error-free environment because these strategies allow the universal defection strategy (i.e., the $T_n$ strategy) to enjoy benefits from free riding on the cooperation of these strategies (see Corollary \ref{stability Tk}).

To see the effect of mistakes on the stability of the $T_k$ strategy, consider a population that entirely composed of players who have  a $T_k$ strategy with the same hardness level $k$. 
A player with the $T_k$ strategy receives $\myFs{C}$ from the first round and maintains cooperation in the next round as long as the number of defections is less than $n-k-1$. Therefore, the mutual cooperation continues when
\begin{itemize}
\item the number of mistakes made by the $n-1$ other conditional cooperators is less than or equal to $n-k-2$ (in this case, the maintenance of mutual cooperation does not depend on whether the focal individual makes a mistake or not), or
\item the number of mistakes made by the $n-1$ other conditional cooperators is exactly equal to $n-k-1$, and the focal player does not make a mistake. 
\end{itemize} 
The first case occurs with probability $\sum_{q=0}^{n-k-2}\psi(n-1,q,\varepsilon)$ and the second case with probability $\psi(n-1, n-k-1, \varepsilon)\times (1-\varepsilon)$. By the same token, the mutual cooperation breaks when 
\begin{itemize}
\item the number of mistakes made by the $n-1$ other conditional cooperators is exactly equal to $n-k-1$, and the focal player makes a mistake, or
\item the number of mistakes made by the $n-1$ other conditional cooperators is greater than $n-k-1$. 
\end{itemize} 
Each case occurs with probability $\psi(n-1, n-k-1, \varepsilon)\times \varepsilon$ and $\sum_{q=n-k}^{n-1}\psi(n-1,q,\varepsilon)$, respectively. Summing up all of the possible cases, the expected payoff for the $T_k$ strategy when the population is entirely composed of individuals following the $T_k$ strategy ($k<n-1$) is written in the following way. 
 
\begin{align}\label {eq:Ta}
\begin{aligned}
\myVs{T_k}=&\myFs{C} + \\
& \delta 
\begin{cases}
\sum_{q=0}^{n-k-2} \myPsi{q} \myVs{T_k} + ~\\
 \myPsi{n-k-1} [ \varepsilon \times 0 + (1-\varepsilon) \times \myVs{T_k} ] + ~ \\
\sum_{q=n-k}^{n-1} \myPsi{q}\times 0, \end{cases}
\end{aligned}
\end{align}
\normalsize
if $k<n-1$.

Suppose that there appears one defector (i.e., a player with strategy $T_n$) in the homogenous population of a $T_k$ strategy where every player has the same tolerance level $k<n-1$. When the game is repeated with probability $\delta$ , the payoff to the $T_n$ player from the repeated game is calculated by
\begin{eqnarray} \label{eq:Tn against k}
\myVs{T_n}&=&\myFs{D}\nonumber\\
&&+\underbrace{\delta\sum_{q=0}^{n-k-2} \myPsi{q}\myVs{T_n}}_{(i)} +\underbrace{\delta\sum_{q=n-k-1}^{n-1} \myPsi{q}\times 0}_{(ii)}.\nonumber\\
\end{eqnarray}
\normalsize
As long as the number of defections, either due to mistakes by the $T_k$ players or the defection of the $T_n$ player, is less than $n-k-1$, $T_k$ players will tolerate defections and  maintain cooperation. In this case, the $T_n$ player keeps enjoying a free rider's advantage from the repeated interaction (see (\textit{i}) in Eq (\ref{eq:Tn against k})). When the number of defections exceeds $n-k-1$, all of the $T_k$ players trigger retaliation, which gives every member a 0 payoff in the following rounds (see (\textit{ii}) in Eq (\ref{eq:Tn against k})). 

It will be shown that it is crucial to have $\myVs{T_k}-\myVs{T_n}>0$ to prove the evolutionarily stability of the $T_k$ strategy. Rearranging Eq (\ref{eq:Ta}) and Eq (\ref{eq:Tn against k}) gives

\begin{eqnarray}
V^\varepsilon(T_k|n-1)=\frac{F^\varepsilon(C|n-1)}{1-\delta\sum_{q=0}^{n-k-2}\psi(n-1, q, \varepsilon)-\delta(1-\varepsilon)\psi(n-1, n-k-1,\varepsilon)}
\end{eqnarray}
and
\begin{eqnarray}
V^\varepsilon(T_n|n-1)=\frac{F^\varepsilon(D|n-1)}{1-\delta\sum_{q=0}^{n-k-2} \psi(n-1, q, \varepsilon)},
\end{eqnarray}
respectively. The condition that the payoff for the $T_k$ strategy is greater than the payoff for the $T_n$ strategy is   
\small{}
\begin{eqnarray}
&&\myVs{T_k}-\myVs{T_n} \nonumber\\
&&= \frac{\myF{C}}{1-\delta\sum_{q=0}^{n-k-2}\myPsi{q}-\delta (1-\varepsilon)\myPsi{n-k-1} }- 
\frac{\myF{D}}{1-\delta\sum_{q=0}^{n-k-2}\myPsi{q} }\nonumber\\
&&= \frac{\myF{D}}{ 1-\delta\sum_{q=0}^{n-k-2}\myPsi{q}-\delta (1-\varepsilon)\myPsi{n-k-1} } \times 
\left[\Delta (\varepsilon;k)-\left( 1- \frac{F^\varepsilon(C|n-1)}{F^\varepsilon(D|n-1)} \right)\right],\nonumber\\
\end{eqnarray}
\normalsize{}where $\Delta(\varepsilon;k)=\frac{\delta(1-\varepsilon)\myPsi{n-k-1}}{1-\delta\sum_{q=0}^{n-k-2}\myPsi{q}}$. Note that the sign of $\myVs{T_k}-\myVs{T_n}$ depends on the sign of $\Delta (\varepsilon;k)-(1-\frac{F^\varepsilon(C|n-1)}{F^\varepsilon(D|n-1)} )$. Because it holds that $0< 1-\frac{F^\varepsilon(C|n-1)}{F^\varepsilon(D|n-1)} <1$, and $\frac{F^\varepsilon(C|n-1)}{F^\varepsilon(D|n-1)}$ is fixed by the payoff structure of the public goods game, it is crucial to have a sufficiently large $\Delta(\varepsilon;k)$.  

A $T_k$ strategy (where $k \in \{ 1,\cdots,n-2 \}$) is evolutionarily stable if 
\begin{eqnarray}
W(T_k|p_k=1)>W(T_{k'}|p_k=1)\quad \forall k'\neq k. 
\end{eqnarray}

In the following proposition, we will show that all of the conditional cooperative strategies are evolutionarily stable when $\varepsilon$ is positive and within a proper open interval and when $\delta$ is sufficiently large. We will show that $W(T_k|p_k=1)>W(T_{k'}|p_k=1)$ for $k \in \{ 1, \ldots, n-2 \}$ in all of the following three cases where (1) a mutant is a conditional cooperator who has a softer strategy (i.e., $k'<k$), (2) a mutant is a conditional cooperator who has a harder strategy (i.e., $k<k'<n$), or (3) a mutant is a defector (i.e., $k'=n$). 

\begin{prop}\label{prop t k}
In an n-person public goods game where the conditional cooperators have a $T_k$ strategy $(k\in \{1, 2, \cdots, n-1\})$, the $T_k$ strategies where $0<k<n-1$ are evolutionarily stable when the probability of game repetition is sufficiently close to 1 and the probability of making mistakes is positive and sufficiently small. 
\end{prop}
\begin{proof}
See Appendix A. 
\end{proof}
\vspace{0.15in}

$T_k$ strategies with ($0<k<n-1$) are not stable at all in an error-free environment but become evolutionarily stable when mistakes are introduced, i.e., they do not allow the invasion of other strategies into their population because the mistakes made by the conditional cooperators produce unintended defections among themselves and leave little room for the defector's free riding. The proof of the above proposition consists of three steps. First, when conditional cooperators make mistakes the expected payoff for the $T_k$ strategy in the repeated game is greater than the expected payoff to the mutant $T_n$ strategy (see Appendix A, Lemma 1). Behavioral mistakes, if they occur with a sufficiently low probability, no longer allow the defectors to invade the population that homogeneously consists of individuals using the $T_k$ strategy (where $0<k<n-1$). Behavioral mistakes produce unintended defections among the cooperators (i.e., even without a defector's invasion), so that softer conditional cooperators are less likely to tolerate a defector's invasion. Second, if the expected payoff for the $T_k$ strategy in the repeated game is greater than that for the mutant $T_n$ strategy, then it is also true that the expected payoff for the $T_k$ strategy in the repeated game is greater than the mutant $T_{k'}$ strategy's payoff if $k'$ is greater than $k$, i.e., if the mutants are harder conditional cooperators (see Appendix A, Lemma 2). Third, behavioral mistakes also make the expected payoff to the $T_k$ strategy in the repeated game when one of these strategies invades $T_k$ population greater than the other conditional cooperative strategies that have$k'$  lower  than $k$ (see Appendix A, Lemma 3). Error makes those strategies distinguishable and makes the payoff for the $T_k$ strategy higher (same logic as proof for Proposition \ref{prop t n-1}).   

In an error-free environment, no conditional cooperative strategies are evolutionarily stable. The $T_{n-1}$ strategy is neutrally stable and is subject to the drift ultimately destroying the cooperative equilibrium, and none of the $T_{k<n-1}$ strategies are stable at all. Proposition \ref{prop t n-1} and  Proposition \ref {prop t k} say that  behavioral errors, if they occur with sufficiently low probability, can make \textit{all} conditional cooperative strategies evolutionarily stable.

\subsection{\label{section selection}Error and selection among evolutionarily stable strategies}

In the above section, we showed that conditional cooperative strategies, $T_k$ for $k \in \{1, \cdots, n-1\}$, are evolutionarily stable when $\varepsilon$ is positive but sufficiently low. Furthermore, each $T_k$ strategy has its own critical value of $\varepsilon_k$ such that the strategy is evolutionary stable in the range of $\varepsilon \in (0, \varepsilon_k)$. If the critical value of each $T_k$ strategy differs according to the hardness level $k$, the magnitude of the probability that a behavioral error occurs could serve as a criterion for equilibrium selection.   

In this section, we will show that the range of $\varepsilon$ supporting the evolutionary stability of the $T_k$ strategy varies according to the tolerance of the conditional cooperative strategies. Let $\varepsilon_k$ be the supremum of $\varepsilon$ that supports the evolutionary stability of a $T_k$ strategy (where $k\in \{1, \cdots, n-1\}$). We will show that $\varepsilon_k$ increases as $k$ decreases, that is, $0<\varepsilon_{n-1}<\varepsilon_{n-2}<\cdots<\varepsilon_1$. For example, if the probability of making an error is sufficiently low so that $0<\varepsilon<\varepsilon_{n-1}$, then all of the conditional cooperative strategies are evolutionarily stable. If the error rate is in the range of $\varepsilon_{n-1}\leq \varepsilon <\varepsilon_{n-2}$, then all of the conditional cooperative strategies except for $T_{n-1}$ are evolutionarily stable, and so on. Lastly, if the error rate is higher than or equal to $\varepsilon_1$, no conditional cooperative strategies are evolutionarily stable, in which case only the universal defection strategy, $T_n$, is evolutionarily stable. In other words, all conditional cooperative strategies are evolutionarily stable when $0<\varepsilon <\varepsilon_{n-1}$ and fewer of the conditional cooperative strategies remain evolutionarily stable as $\varepsilon$ increases. 
\begin{prop}\label{prop_equilibrium selection}
 Let $\varepsilon_k$ be the supremum of $\varepsilon$ that supports the evolutionary stability of the $T_k$ strategy \upshape(\itshape where $k\in \{1, \cdots, n-1\}$\upshape)\itshape. Then, $\varepsilon_k$ increases as $k$ decreases, that is, $0<\varepsilon_{n-1}<\varepsilon_{n-2}<\cdots<\varepsilon_1$.
\end{prop}

\begin{proof}
See Appendix B. 
\end{proof}

In other words, as $\varepsilon$ increases, the $T_k$ strategies with a lower $k$ (i.e., with a higher tolerance level) remain evolutionarily stable. To understand the role of behavioral mistakes in the equilibrium selection in the above proposition, we need to examine the two effects that the mistakes produce. First, the mistakes lower the probability that the mutually cooperative phase continues because they trigger retaliation towards cooperators' unintended defections. Secondly, they also reduce the possibility that the defectors enjoy benefits from free riding on tolerant cooperators. If the incumbent conditional cooperators have a lower $k$, the first effect is less detrimental because the cooperative equilibrium based on the lower $k$ is less vulnerable to a breakdown when a mistake occurs with higher probability.

\section{The General Case Where $\delta$ is Sufficiently High But Not at the Limit of $1$}\label{sec general case}
One should note that only the supremum of $\varepsilon$ matters in Proposition \ref{prop_equilibrium selection}. We proved the proposition at the $\delta=1$ limit, in which case the infimum becomes meaningless as long as the rate of mistakes is positive. At the $\delta=1$ limit, error will at some point terminate the game and enable the $T_k$ strategy with a low $k$ to block the defectors from endlessly accumulating a free rider's benefit. That is, at the $\delta=1$ limit, error appears to work as long as $\varepsilon$ is positive (See Figure \ref{figure} (a)).  

\begin{figure}[p]
\centering
\resizebox{90mm}{!}{\includegraphics{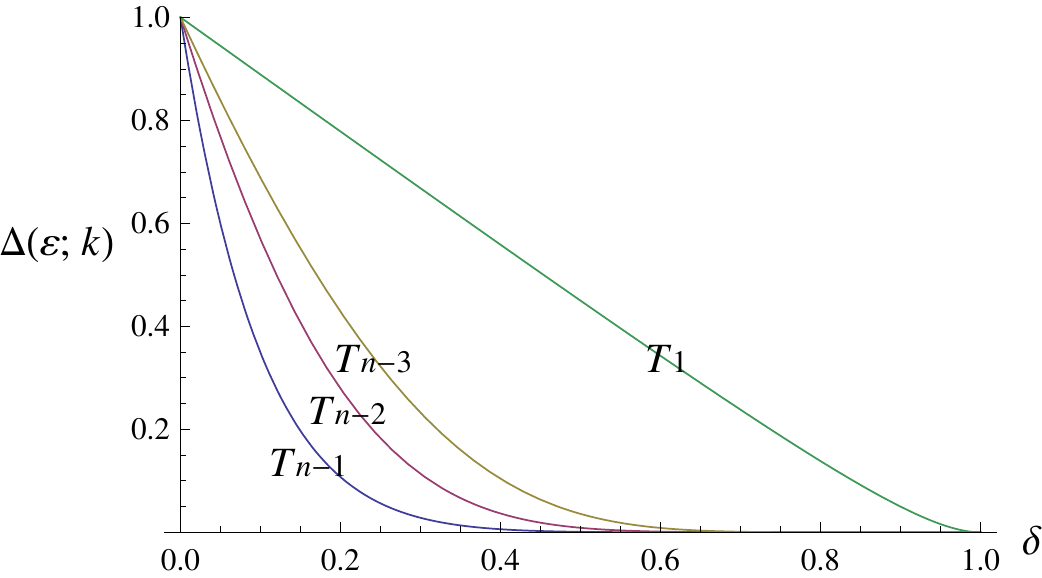}}\\
\vspace{-0.1in}(a) $\delta=1$\\
\vspace{0.25in}
\resizebox{90mm}{!}{\includegraphics{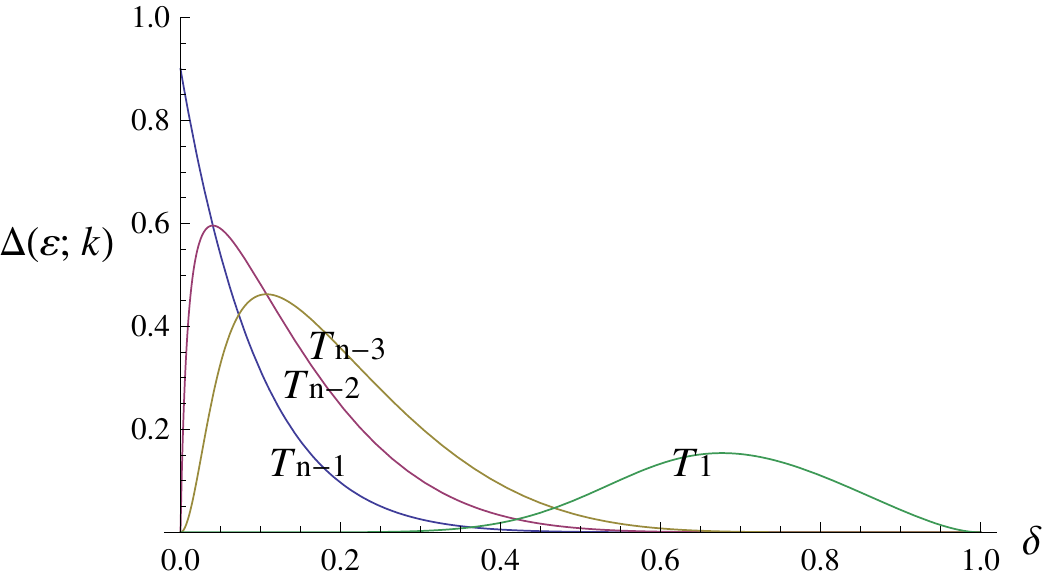}}\\
\vspace{-0.1in}(b) $\delta=0.9$\\
\vspace{0.25in}
\resizebox{90mm}{!}{\includegraphics{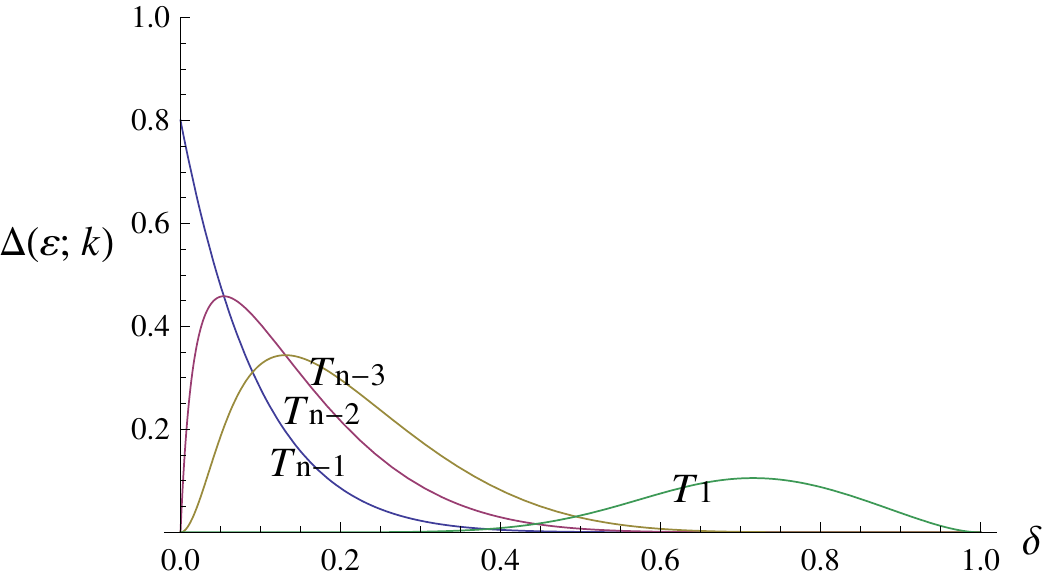}}\\
\vspace{-0.1in}(c) $\delta=0.8$
\caption{$\Delta(\varepsilon; k)$ with $\delta$. ($n=10$)}\label{figure}
\end{figure}

However, when $\delta$ is large but not at the limit of 1, the $T_k$ strategy ($0<k<n-1$) needs a sufficiently high rate of $\varepsilon$ to block the defectors' free riding. The role of error in supporting the evolutionary stability of cooperative strategies depends on whether it produces enough defections among the cooperators before the defectors' invasion. Note that the number of mistakes to prevent defectors from invading depends on the level of tolerance. As $k$ becomes lower, more unintended defections among the cooperators are needed to terminate the game repetition (i.e., to block a defector from gaining a free riding benefit) in case one defector appears in a group. In other words, there should be an infimum of $\varepsilon$ to make the $T_k$ strategy evolutionarily stable. In this section, we show that the conclusions derived in the former section are still valid with some modification, especially for the infimum in the general case where $\delta$ is sufficiently high but not at the limit of 1.    

Here, we analyze a more general case where $\delta$ is less than 1. At the $\delta=1$ limit, we showed that each $T_k$ strategy, where $k<n$, has the supremum, $\varepsilon_k$, such that the $T_k$ strategy is evolutionary stable for $\varepsilon \in (0, \varepsilon_k)$ and that $\varepsilon_k$ is increasing as $k$ decreases.

In the following propositions, we will now show that each conditional cooperative $T_k$ strategy has its own band of error rate $(\underline{\varepsilon}_k, \overline{\varepsilon}_k$) that makes the $T_k$ strategy evolutionarily stable. In other words, there appears an infimum of $\varepsilon$ that supports the evolutionary stability of each $T_k$ strategy and the infimum moves toward zero as $\delta$ approaches 1. Now we can provide a generalized versions of Proposition \ref{prop t n-1}, Proposition \ref{prop t k} and Proposition \ref{prop_equilibrium selection}. 

\begin{prop}\label{general stability}
In an n-person public goods game where the conditional cooperators have a $T_k$ strategy $(k\in \{1, 2, \cdots, n-1\})$, all of the $T_k$ strategies are evolutionarily stable when the probability of game repetition is sufficiently high and the probability of making mistakes is in the range of \upshape($\underline{\varepsilon}_k, \overline{\varepsilon}_k$).  
\end{prop}
\begin{proof}
See Appendix C. 
\end{proof}

\begin{prop}\label{prop_general equilibrium selection}
In an n-person public goods game where conditional cooperators have a $T_k$ strategy $(k\in \{1, 2, \cdots, n-1\})$, 
both $\underline{\varepsilon}_k$ and $\overline{\varepsilon}_k$ increase as $k$ decreases and $\underline{\varepsilon}_{n-1}=0$, that is, $0=\underline{\varepsilon}_{n-1} < \underline{\varepsilon}_{n-2} < \cdots<\underline{\varepsilon}_1$ and $0<\overline{\varepsilon}_{n-1} < \overline{\varepsilon}_{n-2} < \cdots<\overline{\varepsilon}_1 < 1$, when the probability of game repetition is sufficiently high. 
\end{prop}
\begin{proof}
See Appendix D. 
\end{proof}

Figure \ref{figure} makes our argument clear. All of the above propositions show that the stability condition critically depends on the sign of $\Delta(\varepsilon; k)-(1-\frac{\myFs{C}}{\myFs{D}})$. This means that $\Delta(\varepsilon;k)$ should be sufficiently large to offset  $1-\frac{\myFs{C}}{\myFs{D}}$. The three panels in Figure \ref{figure} show $\Delta(\varepsilon; k)$ for three different values of $\delta$. One can easily check the following. First, the $\Delta(\varepsilon; k)$ value becomes larger as $\delta$ increases, and our previous propositions hold if $\delta$ approaches 1. Second, when $\delta$ is less than one, the probability of mistakes should be sufficiently large (i.e., the infimum is important as well as its supremum to support cooperation) and the infimum becomes zero when $\delta$ approaches 1.

\section{Conclusion}

This study examined the effect of behavioral mistakes on the dynamic stability of the cooperative equilibrium in a repeated public goods game. This study shows that while a behavioral mistake has detrimental effects on cooperation because it reduces the length of the period of mutual cooperation by triggering the conditional cooperators' retaliation, a behavioral mistake also has a positive effect by making the conditional cooperative strategies evolutionarily stable. This paper shows that the behavioral mistakes stabilize the cooperative equilibrium based on the hardest cooperative strategy by eliminating the behavioral indistinguishability between the conditional cooperative strategies at the cooperative equilibrium. This paper shows that the mistakes also stabilize the cooperative equilibrium based on a softer cooperative strategy by producing unintended defection among the cooperators and making the softer conditional cooperators less tolerant of a defector's invasion. Finally, the paper shows that the error rate, $\varepsilon$, could serve as an equilibrium selection criterion because each equilibrium that is based on a different level of tolerance is supported by different ranges for the error rate. 

\appendix


\section{The Proof of Proposition 3}\label{proof_Prop 3}
Consider a population that is entirely composed of conditional cooperators who have the same hardness (where $k<n-1$). The payoff to the $T_k$ strategy when $p_k=1$ is
\small{}
\begin{eqnarray}
W(T_k|p_k=1)&=&V^\varepsilon(T_k|n-1)m(n-1,p_k=1)\nonumber\\ 
&=&F^\varepsilon(C|n-1)\nonumber\\
&&+\delta \left[\sum_{q=0}^{n-k-2}\psi(n-1, q, \varepsilon)+(1-\varepsilon)\psi(n-1, n-k-1,\varepsilon)\right]V^\varepsilon(T_k|n-1),\nonumber\\
\end{eqnarray}
\normalsize
which gives
\small{}
\begin{eqnarray}
W(T_k|p_k=1)&=&V^\varepsilon(T_k|n-1)\nonumber\\
&=&\frac{F^\varepsilon(C|n-1)}{1-\delta\sum_{q=0}^{n-k-2}\psi(n-1, q, \varepsilon)-\delta(1-\varepsilon)\psi(n-1, n-k-1,\varepsilon)},\nonumber\\
\end{eqnarray}
\normalsize{}because $m(n-1, p_k=1)=1$. The payoff to the mutant $T_n$ strategy is
\small{}
\begin{eqnarray}
W(T_n|p_k=1)&=&V^\varepsilon(T_n|n-1)m(n-1,p_k=1)\nonumber\\ 
&=&F^\varepsilon(D|n-1)+\delta\sum_{q=0}^{n-k-2}\psi(n-1, q, \varepsilon) V^\varepsilon(T_n|n-1),
\end{eqnarray}
\normalsize{}which gives
\small{}
\begin{eqnarray}
W(T_n|p_k=1)=V^\varepsilon(T_n|n-1)=\frac{F^\varepsilon(D|n-1)}{1-\delta\sum_{q=0}^{n-k-2} \psi(n-1, q, \varepsilon)}.
\end{eqnarray}
\normalsize

Now, the following three lemmas show that any $T_k$ (for $ 0<k<  n-1)$ is evolutionarily stable for a sufficiently large $\delta$ close to $1$ and a sufficiently small $\varepsilon \in (0,1)$. The proof is presented at the end.

\begin{lemma}\label{lem:one}
For a sufficiently large $\delta$ close to $1$, there exists a unique $\varepsilon_k$ such that $\myVs{T_k}-\myVs{T_n}>0$ for $ \varepsilon \in (0,\varepsilon_k)$.
\end{lemma}

\begin{proof}
The sign of $\myVs{T_k}-\myVs{T_n}$ is determined by $\Delta(\varepsilon;k)-(1-\frac{F^\varepsilon(C|n-1)}{F^\varepsilon(D|n-1)})$
where $\Delta(\varepsilon;k)=\frac{\delta(1-\varepsilon)\myPsi{n-k-1}}{1-\delta\sum_{q=0}^{n-k-2}\myPsi{q}}$. When $\delta$ is close to 1, we have
\small{}
\begin{align}
\lim_{\delta\rightarrow 1} \Delta(\varepsilon;k)=&\frac{(1-\varepsilon)\myPsi{n-k-1}}{1-\sum_{q=0}^{n-k-2}\myPsi{q}}\nonumber\\
=&\frac{(1-\varepsilon)\psi(n-1, n-k-1, \varepsilon)}{\sum_{q=n-k-1}^{n-1}\psi(n-1, q, \varepsilon)}\nonumber\\
=& \frac{(1-\varepsilon)\binom{n-1}{n-k-1}}{\sum_{q=0}^k \binom{n-1}{n-k-1+q}\left(\frac{\varepsilon}{1-\varepsilon}\right)^q}
\end{align}
\normalsize{}and 
\small{}
\begin{align*}
&\lim_{\varepsilon\rightarrow 0}\lim_{\delta\rightarrow 1} \Delta(\varepsilon;k)=1 \\
&\lim_{\varepsilon\rightarrow 1}\lim_{\delta\rightarrow 1} \Delta(\varepsilon;k)=0.
\end{align*}
\normalsize{}Because $0<1-\frac{F^\varepsilon(C|n-1)}{F^\varepsilon(D|n-1)}<1$, 
\small{}
\begin{align*}
\begin{cases}
\myVs{T_k}-\myVs{T_n}>0  & \text{if $\varepsilon \to 0$}\\
\myVs{T_k}-\myVs{T_n}<0 & \text{if $\varepsilon \to 1$}.
\end{cases}
\end{align*}
\normalsize{}Because $V^\varepsilon(\cdot)$ is continuous for $\varepsilon \in (0,1)$, then there exists at least one $\varepsilon$ that makes $\myVs{T_k} = \myVs{T_n}$. Now, to prove the uniqueness of this $\varepsilon$, $\varepsilon_k$, we take first derivative of $\lim_{\delta\rightarrow 1}\Delta(\varepsilon, k)$:
\small{}
\begin{align*}
\hspace{-0.5em}
&\frac{\mathrm{d} \lim_{\delta\rightarrow 1}\Delta}{\mathrm{d} \varepsilon} = \\
&\frac{ -\binom{n-1}{n-k-1}\sum_{q=0}^{k} \binom{n-1}{n-k-1+q} \left( \frac{\varepsilon}{1-\varepsilon} \right)^q-(1-\varepsilon)\binom{n-1}{n-k-1} \sum_{q=0}^{k} q \binom{n-1}{n-k-1+q} \left( \frac{1}{(1-\varepsilon)^2} \right)^{q-1}}{\left( \sum_{q=0}^{k} \binom{n-1}{n-k-1+q} \left( \frac{\varepsilon}{1-\varepsilon} \right)^q \right)^2}
\end{align*}
\normalsize
which is always negative for $\varepsilon \in (0,1)$.
\end{proof}

\begin{lemma}
Suppose that a mutant $T_{k'}$ \upshape{(}\itshape for $k' \in \{k+1,\ldots,n-1 \}$ \upshape{)} \itshape  appears in a homogeneous population of $T_k$.  $\myVs{T_{k}}-\myVs{T_{{k'}}}$ has the same sign as $\myVs{T_k}-\myVs{T_n}$. 
\end{lemma} 

\begin{proof}
The payoff for strategy $T_{a}$ from the repeated game is given by  
\small{}
\begin{align*}
\myVs{T_{k'}} = &\myFs{C}+ \delta
\begin{cases}
\sum_{q=0}^{n-k'-1}\myPsi{q} \cdot \myVs{T_{k'}} + ~ \\
\sum_{q=n-k'}^{n-k-2}\myPsi{q} \cdot \myVs{T_{n}} + ~ \\
\sum_{q=n-k-1}^{n-1}\myPsi{q} \cdot 0.
\end{cases}
\end{align*}
\normalsize{}To check the stability, it is noted that 
\footnotesize{}
\begin{align*}
&\myV{T_k} -\myV{T_{k'}} = \\
& \frac{ (\delta \sum_{q=n-k'}^{n-k-2}\myPsi{q} + \delta(1-\varepsilon)\myPsi{n-k'-1} ) \myF{D} }{(1-\delta\sum_{q=0}^{n-k-2}\myPsi{q}-\delta(1-\varepsilon)\myPsi{n-k-1}) (1-\delta \sum_{q=0}^{n-k'-1} \myPsi{q} ) } \times \\
& \left[
\frac{F^\varepsilon(C|n-1)}{F^\varepsilon(D|n-1)}-\underbrace{ \frac{\delta \sum_{q=n-k'}^{n-k-2} \myPsi{q} \{ 1-\delta \sum_{q=0}^{n-k-2}\myPsi{q}-\delta(1-\varepsilon)\myPsi{n-k-1} \} }{ (1-\delta \sum_{q=0}^{n-k-2}\myPsi{q}) \{ \delta \sum_{q=n-k'}^{n-k-2} \myPsi{q} + \delta (1-\varepsilon)\myPsi{n-k-1} \} }%
}_{(**)}
\right],
\end{align*}
\normalsize{}and 
\small{}
\begin{align*}
(**) \leq \left[ 1-\dfrac{\delta (1-\varepsilon) \myPsi{n-k-1} }{1-\delta\sum_{q=0}^{n-k-2} \myPsi{q} } \right].
\end{align*}
\normalsize{}The RHS of the last inequality is the exactly the same condition for $\myVs{T_k}-\myVs{T_n}>0$, which completes the proof.
\end{proof}

\begin{lemma}
When a mutant $T_{k'}$ \upshape{(}\itshape for $k' <k$\upshape{)} \itshape appears in the homogeneous population of $T_k$, $\myVs{T_k}-\myVs{T_{k'}}>0$.
\end{lemma}

\begin{proof}
It is easy to see that the $T_{k}$ is strictly dominant against a single mutant $T_{k'}$ for all parameters because $\frac{b}{n}-c < 0$ (See also the proof of Proposition 2).
\end{proof}

The proof for Proposition 3 is as follows:
\begin{proof}
According to Lemma 1, there exists an $\varepsilon_k$ such that $\myVs{T_k}-\myVs{T_n}> 0$ if $\varepsilon < \varepsilon_k$ and  $\myVs{T_k}-\myVs{T_n} \leq 0$ otherwise. According to Lemma 2, $\myVs{T_k}-\myVs{T_{k'}}>0$ for $k<k'<n$ if $\varepsilon$ satisfies $\myVs{T_k}-\myVs{T_n}\geq 0$. Lastly, according to Lemma 3, $\myVs{T_k}$ is always greater than $\myVs{T_{k'}}$ if $k'<k$. These three lemmas lead to the conclusion that $T_k$ is evolutionarily stable if $\varepsilon \in (0, \varepsilon_k)$.
\end{proof}

\section{The Proof of Proposition 4}

\begin{proof}
According to Lemma 1, because $\frac{F^\varepsilon(C|n-1)}{F^\varepsilon(D|n-1)}-1$ is fixed, $\Delta(\varepsilon;k)$ determines $\myVs{T_{k}} -\myVs{T_n}$. Thus, when $\Delta(\varepsilon;k) > \Delta(\varepsilon;k+1)$ holds for $\varepsilon \in (0,1)$, the proof is completed. 
\small{}
\begin{align*}
\Delta(\varepsilon;k+1)-&\Delta(\varepsilon;k) = \\
& (1-\varepsilon) \underbrace{\left[%
\frac{\myPsi{n-k-2}}{1-\sum_{q=0}^{n-k-3}\myPsi{q}}-\frac{\myPsi{n-k-1}}{1-\sum_{q=0}^{n-k-2}\myPsi{q} }%
\right]}_{(***)}
\end{align*}
\normalsize{}$(***)$ can be rearranged to
\small{}
\begin{align*}
\frac{1}{1 + \frac{ \binom{n-1}{n-k-1} }{\binom{n-1}{n-k-2}}+\cdots+\frac{ \binom{n-1}{n-1} }{\binom{n-1}{n-k-2}}}
-
\frac{1}{1 + \frac{ \binom{n-1}{n-k} }{\binom{n-1}{n-k-1}}+\cdots+\frac{ \binom{n-1}{n-1} }{\binom{n-1}{n-k-1}}}.
\end{align*}
\normalsize{}It is easy to check that the denominator of the first part is always larger than that of the second, which confirms that $\Delta(\varepsilon;k) > \Delta(\varepsilon;k+1)$. In other words, the $T_n$ strategy can always invade a population that is entirely composed of individuals using $T_{k+1}$ strategy if $\varepsilon=\varepsilon_k$.  Following the previous lemmas, there should be an $\varepsilon_{k+1}$ such that the $T_{k+1}$ strategy is evolutionarily stable for $\varepsilon \in (0,\varepsilon_{k+1})$. $\myVs{T_{k + 1}}-\myVs{T_n} > 0$ for $\varepsilon \in (0, \varepsilon_{k+1})$ and $\myVs{T_{k + 1}}-\myVs{T_n} < 0$ for $\varepsilon=\varepsilon_k$, which implies that $\varepsilon_{k+1}< \varepsilon_k$ for  $k \in \{ 1, \ldots, n-2 \}$.
\end{proof}

\section{The Proof of Proposition 5}

In the proof of Proposition \ref{prop t k} (Appendix A), only Lemma 1 needs the necessary condition that $\delta$ is at the limit of $1$. We will show that $\myVs{T_k}-\myVs{T_n} > 0$ for $\varepsilon \in (\underline{\varepsilon_k},\overline{\varepsilon_k})$ as long as $\delta$ is sufficiently high such that $\frac{\myF{C}}{\myF{D}}-1+\Delta(\varepsilon;k)>0$ but is not close to $1$. We will provide three lemmas and present the proof of the proposition at the end. 

First of all, $\Delta(\varepsilon;k)$ is rewritten as 
\small{}
\begin{align*}
\Delta(\varepsilon;k)=\frac{\delta(1-\varepsilon)\myPsi{n-k-1} }{1-\delta \sum_{q=0}^{n-k-2}\myPsi{q} } =\frac{ \delta }{D_1(\varepsilon;k)+D_2(\varepsilon;k)},
\end{align*}
\normalsize{}where 
\small{}
\begin{align*}
D_1(\varepsilon;k)=\frac{1-\delta}{\binom{n-1}{n-k-1}}(\varepsilon)^{-n+k+1}(1-\varepsilon)^{-k-1},~%
D_2(\varepsilon;k)=\delta \sum_{q=0}^{k} \frac{\binom{n-1}{n-k-1+q}}{\binom{n-1}{n-k-1}}(\varepsilon)^{q}(1-\varepsilon)^{-q-1}.
\end{align*} 
\normalsize
$D_1$ and $D_2$ are obtained by 
\small
\begin{align*}
1-\delta \sum_{q=0}^{n-k-2}\myPsi{q} & = \left( 1-\sum_{q=0}^{n-k-2}\myPsi{q} \right) + \sum_{q=0}^{n-k-2}\myPsi{q}(1-\delta ) \\
& = \sum_{q=n-k-1}^{n-1}\myPsi{q} + \left(1-\sum_{q=n-k-1}^{n-1}\myPsi{q}\right)(1-\delta) \\
& = (1-\delta) + \delta \sum_{q=n-k-1}^{n-1}\myPsi{q}.
\end{align*}
\normalsize
Now, our discussion is based on $D_1(\varepsilon;\cdot)+D_2(\varepsilon;\cdot)$ instead of $\Delta(\varepsilon;\cdot)$; because $D_1,D_2>0$, this can be one-to-one mapped inversely onto $\Delta$.

\begin{lemma}\label{lem_prop_5_1}
$D_1(\varepsilon;k)+D_2(\varepsilon;k)$ is strictly convex for $\varepsilon \in (0,1)$. 
\end{lemma}
\begin{proof}
It is to be shown that $D_1^{''}(\varepsilon,\cdot)>0$ and $D_2^{''}(\varepsilon,\cdot)>0$ for $\varepsilon \in (0,1)$. A direct calculation shows that 
\small
\begin{align*}
&\frac{\partial^2 [(\varepsilon)^{-n+k+1}(1-\varepsilon)^{-k-1}]}{\partial \varepsilon^2} = \\
&~~~~~~~~\left[ n(1+n)\varepsilon^2-2(n+1)(n-k-1) \varepsilon + (n-k)(n-k-1) \right](1-\varepsilon )^{-k-3} \varepsilon
   ^{k-n-1} > 0 \\
&\frac{\partial^2 [(\varepsilon)^q(1-\varepsilon)^{-q-1}]}{\partial \varepsilon^2} =  [2\varepsilon^2 + 4 \varepsilon q + q(q-1)] (\varepsilon)^{q-2}(1-\varepsilon)^{-q-2} > 0.
\end{align*}
\normalsize
for $\varepsilon \in (0,1)$.
\end{proof}

\begin{lemma}\label{lem_prop_5_2}
For $k\in \{1, 2, \cdots, n-2\}$, there exists at least one $\varepsilon_k^* \in (0,1)$ such that $D_1^{'}(\varepsilon_k^*,k) + D_2^{'}(\varepsilon_k^*,k) = 0$. For $k=n-1$, no such $\varepsilon_k^*$ exists. 
\end{lemma}
\begin{proof}
First of all, it is given that $D_1^{'}(\varepsilon,n-1)+D_2^{'}(\varepsilon,n-1) =\frac{1}{\delta(1-\varepsilon)^n}$, which shows that no $\varepsilon_k^*$ exists. A direct calculation shows that 
\small
\begin{align}
\frac{\partial [(\varepsilon)^{-n+k+1}(1-\varepsilon)^{-k-1}]}{\partial \varepsilon} &= [k+1-(1-\varepsilon)n](\varepsilon)^{k-n}(1-\varepsilon)^{-2-k} \label{eq:appdx:01} \\
\frac{\partial [(\varepsilon)^q(1-\varepsilon)^{-q-1}]}{\partial \varepsilon} &=  q (\varepsilon)^{q-1}(1-\varepsilon)^{-q-1} + (q+1)(\varepsilon)^q (1-\varepsilon)^{-q-2} \label{eq:appdx:02}.
\end{align}
\normalsize
For a sufficiently small $\varepsilon$, (\ref{eq:appdx:01}) is negative, and its absolute value can be arbitrarily larger, but (\ref{eq:appdx:02}) is positive, and its value can be arbitrarily smaller. Then, we have  $D_1^{'}(\varepsilon_k^*,k)+D_2^{'}(\varepsilon_k^*,k) < 0$ for a sufficiently small $\varepsilon$. For a sufficiently large $\varepsilon$, it is easy to check that $D_1^{'}(\varepsilon_k^*,k)+D_2^{'}(\varepsilon_k^*,k) > 0$. As $D_1'(\cdot)$ and $D_2'(\cdot)$ are continuous, there exists at least one $\varepsilon_k^* \in (0,1)$ that makes $D_1^{'}(\varepsilon_k^*,k)+D_2^{'}(\varepsilon_k^*,k) = 0$. 
\end{proof}
\begin{lemma}\label{lem_prop_5_3}
For a given $k\in \{1, 2, \cdots, n-2\}$, $D_1(\varepsilon,k)+D_2(\varepsilon,k)$ has a unique minimum over $\varepsilon \in (0,1)$. The infimum of $D_1(\varepsilon,n-1)+D_2(\varepsilon,n-1)$ is obtained at $\varepsilon=0$.
\end{lemma}
\begin{proof}
It is trivial for the case of $k=n-1$. For $0<k<n-1$, the previous two lemmas are the conditions that $D_1(\varepsilon,k)+D_2(\varepsilon,k)$ has a unique minimum somewhere at $\varepsilon \in (0,1)$. 
\end{proof}

The proof for Proposition 5 is as follows:
\begin{proof}
At first, it is easy to check that $k=n-1$ makes $\underline{\varepsilon}_{n-1}=0$, and the proof can be provided by the same method used in Proposition 3. Lemma \ref{lem_prop_5_3} implies that $\Delta(\varepsilon,k)$ should have a single maximum over $\varepsilon \in (0,1)$ for $0<k<n-1$. When the proper $\delta$ is given, there exists an $\varepsilon \in (\underline{\varepsilon}_k, \overline{\varepsilon}_k)$ that makes $\myVs{T_k}-\myVs{T_n} > 0$. 
\end{proof}

\section{The Proof of Proposition 6}

For the proof, we need a characterization of the shapes of $D_1(\varepsilon;k)+D_2(\varepsilon;k)$ over $\varepsilon \in (0,1)$ for the different $k$. For this characterization, the following three lemmas are provided, and the proof is presented at the end. 

\begin{lemma}\label{lem_prop_6_1}
For any $k\in \{1, 2, \cdots, n-2\}$, $D_1^{'}(\varepsilon;k+1) + D_2^{'}(\varepsilon;k+1) > D_1^{'}(\varepsilon;k) + D_2^{'}(\varepsilon;k)$ over $\varepsilon \in (0,1)$.
\end{lemma}

\begin{proof}
A direct calculation shows that 
\small
\begin{align*}
&D_1^{\rq{}}(\varepsilon;k+1)-D_1^{\rq{}}(\varepsilon;k)=\\
&~~\frac{(1-\delta) (1-\varepsilon)^{-k-3} \varepsilon^{k-n}[ \varepsilon\binom{n-1}{-k+n-1}(k+n (\varepsilon-1)+2)-(1-\varepsilon) \binom{n-1}{-k+n-2} (k+n(\varepsilon -1)+1)]}{\binom{n-1}{-k+n-2} \binom{n-1}{-k+n-1}}>0,\\
&D_2^{\rq{}}(\varepsilon;k+1)-D_2^{\rq{}}(\varepsilon;k)=\\
&~~\delta\sum_{q=0}^{k}\left[ \left( \frac{nq}{(k+1)(n-k-1)(k-q+1)(n+q-k-1)}\right) \left\{ q(\varepsilon)^{q-1}(1-\varepsilon)^{-q-1} + (q+1)(\varepsilon)^q (1-\varepsilon)^{-q-2} \right\} \right]\\
&~~ +\frac{\delta}{\binom{n-1}{n-k-2}}[q\varepsilon^{k}(1-\varepsilon)^{-k-2}+(q+1)\varepsilon^{k+1}(1-\varepsilon)^{-k-3}]>0 
\end{align*}
\normalsize{}
\end{proof}

Let us define $\varepsilon^*_{k}=  \text{\upshape arg min} \{D_1(\varepsilon;k) + D_2(\varepsilon;k): \varepsilon \in (0,1) \}$ for $k\in \{1, 2, \cdots, n-1\}$.

\begin{lemma}\label{lem_prop_6_2} 
For any $0<k<n-1$,
$D_1(\varepsilon_k^*;k+1)+D_2(\varepsilon_k^*;k+1) = D_1(\varepsilon_k^*;k)+D_2(\varepsilon_k^*;k)$ holds.
\end{lemma}

\begin{proof}
$\varepsilon^*_k$ satisfies $D_1^{'}(\varepsilon_k^*;k)+D_2^{'}(\varepsilon_k^*;k)=0$. We should have $[D_1(\varepsilon_k^*;k)-D_1(\varepsilon_k^*;k+1)]+ [D_2(\varepsilon_k^*;k)-D_2(\varepsilon_k^*;k+1)]=0$ for the proof. A direct calculation along with $D_1^{'}(\varepsilon_k^*;k)+D_2^{'}(\varepsilon_k^*;k)=0$ shows that $-\frac{\varepsilon_k^*}{n-k-1} D_1^{'}(\varepsilon_k^*,k) = [D_1(\varepsilon_k^*;k)-D_1(\varepsilon_k^*;k+1)]=\frac{\varepsilon_k^*}{n-k-1} D_2^{'}(\varepsilon_k^*,k)$. Some calculation shows that $\frac{\varepsilon_k^*}{n-k-1} D_2^{'}(\varepsilon_k^*,k)+[D_2(\varepsilon_k^*;k)-D_2(\varepsilon_k^*;k+1)]=0$. 
\end{proof}
\begin{lemma}\label{lem_prop_6_3}
For $0<k<n-1$, 
\small
\begin{align*}
\begin{cases}
D_1(\varepsilon;k)+D_2(\varepsilon;k) > D_1(\varepsilon;k+1)+D_2(\varepsilon;k+1) & \text{for $\varepsilon \in (0,\varepsilon_k^*)$} \\
D_1(\varepsilon;k)+D_2(\varepsilon;k) < D_1(\varepsilon;k+1)+D_2(\varepsilon;k+1) & \text{for $\varepsilon \in (\varepsilon_k^*,1)$}.
\end{cases}
\end{align*}
\end{lemma}
\normalsize
\begin{proof} 
Let us define that
\small
\begin{align*}
d(\varepsilon)=\left[D_1(\varepsilon;k)+D_2(\varepsilon;k)\right]- \left[D_1(\varepsilon;k+1)+D_2(\varepsilon;k+1)\right].
\end{align*}
\normalsize{}It is noted that $d'(\varepsilon)<0$ by Lemma \ref{lem_prop_6_1}, and $d(\varepsilon_k^*)=0$ by Lemma \ref{lem_prop_6_2}. As two curves cannot be tangent, $d(\varepsilon)$ should change its sign around $\varepsilon_k^*$. We prove that $d(\varepsilon)$ should change its sign only once around $\varepsilon_k^*$, from positive to negative. 
\begin{inparaenum}[i)]
\item For a sufficiently small value of $\varepsilon$, $D_1(\varepsilon;k)-D_1(\varepsilon;k+1)$ can become arbitrarily larger while $D_2(\varepsilon;k+1)-D_2(\varepsilon;k)$ can become arbitrarily smaller. Hence, there exists an $\hat{\varepsilon}_k$ such that $d(\varepsilon)>0$ is ensured for $\varepsilon \in (0,\hat{\varepsilon}_k)$.
\item Let us assume that $d(\varepsilon)$ would change its sign around $\hat{\varepsilon}_k < \varepsilon_k^*$ from positive to negative. As $d(\hat\varepsilon_k)=0$ and $d'(\varepsilon)<0$, $d(\varepsilon)<0$ holds for $\varepsilon \in (\hat{\varepsilon}_k,1)$. Thus, $d(\varepsilon_k^*)=0$ cannot be satisfied without violating the strict convexity of $D_1 + D_2$.  
\item Once $d(\varepsilon_k^*)=0$ is realized, by the same logic, $d(\varepsilon)<0$ holds for $\varepsilon \in (\varepsilon_k^*,1)$. 
\end{inparaenum}
\end{proof}
The proof for Proposition 6 is as follows:
\begin{proof}
According to Lemma \ref{lem_prop_6_3}, for $0 < \varepsilon < \varepsilon^*_k$, $\Delta(\varepsilon;k)< \Delta(\varepsilon;k+1)$ holds, and $\Delta(\varepsilon;k)> \Delta(\varepsilon;k+1)$ holds for $\varepsilon^*_k<\varepsilon<1$. Also the lemma implies that the single-maximum curves of $\Delta(\varepsilon;k)$ and $\Delta(\varepsilon;k+1)$ should be configured such that the only intersection between the two is realized at the descending part of $\Delta(\varepsilon;k+1)$. When a sufficiently high $\delta$ is given, the proposition is satisfied.
\end{proof}

\bibliography{choihuhh}

\begin{thebibliography}{19}
\expandafter\ifx\csname natexlab\endcsname\relax\def\natexlab#1{#1}\fi
\expandafter\ifx\csname url\endcsname\relax
  \def\url#1{{\tt #1}}\fi
\expandafter\ifx\csname urlprefix\endcsname\relax\def\urlprefix{URL }\fi

\bibitem[{Axelrod(1997)}]{Axelrod:1997}
Axelrod, R. (1997).
\newblock {\em The Complexity of Cooperation: {Agent-Based} Models of
  Competition and Collaboration\/}.
\newblock Princeton University Press, 1st printing ed.

\bibitem[{Axelrod \& Hamilton(1981)}]{Axelrod_Hamilton:1991}
Axelrod, R., \& Hamilton, W.~D. (1981).
\newblock The evolution of cooperation.
\newblock {\em Science\/}, {\em 211\/}(4489), 1390--1396.

\bibitem[{Bendor \& Mookherjee(1987)}]{Bendor_Mookherjee:1987}
Bendor, J., \& Mookherjee, D. (1987).
\newblock Institutional structure and the logic of ongoing collective action.
\newblock {\em The American Political Science Review\/}, {\em 81\/}(1),
  129--154.

\bibitem[{Binmore \& Samuelson(1992)}]{Binmore_Samuelson:1992}
Binmore, K., \& Samuelson, L. (1992).
\newblock Evolutionary stability in repeated games played by finite automata.
\newblock {\em Journal of Economic Theory\/}, {\em 57\/}(2), 278--305.

\bibitem[{Boyd(1989)}]{Boyd:1989}
Boyd, R. (1989).
\newblock Mistakes allow evolutionary stability in the repeated prisoner's
  dilemma game.
\newblock {\em Journal of Theoretical Biology\/}, {\em 136\/}(1), 47--56.

\bibitem[{Boyd \& Lorberbaum(1987)}]{Boyd_Lorberbaum:1987}
Boyd, R., \& Lorberbaum, J.~P. (1987).
\newblock No pure strategy is evolutionarily stable in the repeated prisoner's
  dilemma game.
\newblock {\em Nature\/}, {\em 327\/}(6117), 58--59.

\bibitem[{Boyd \& Richerson(1988)}]{Boyd_Richerson:1988}
Boyd, R., \& Richerson, P. (1988).
\newblock The evolution of reciprocity in sizable groups.
\newblock {\em Journal of Theoretical Biology\/}, {\em 132\/}(3), 337--356.

\bibitem[{Choi(2007)}]{Choi:2007}
Choi, J.-K. (2007).
\newblock Trembles may support cooperation in a repeated prisoner's dilemma
  game.
\newblock {\em Journal of Economic Behavior \& Organization\/}, {\em 63\/}(3),
  384--393.

\bibitem[{Farrell \& Ware(1989)}]{Farrell_Ware:1989}
Farrell, J., \& Ware, R. (1989).
\newblock Evolutionary stability in the repeated prisoner's dilemma.
\newblock {\em Theoretical Population Biology\/}, {\em 36\/}(2), 161--166.

\bibitem[{Fudenberg \& Maskin(1986)}]{Fudenberg_Maskin:1986}
Fudenberg, D., \& Maskin, E. (1986).
\newblock The folk theorem in repeated games with discounting or with
  incomplete information.
\newblock {\em Econometrica\/}, {\em 54\/}(3), 533--554.

\bibitem[{Fundenberg \& Maskin(1990)}]{Fudenberg_Maskin:1990}
Fundenberg, D., \& Maskin, E. (1990).
\newblock Evolution and cooperation in noisy repeated games.
\newblock {\em The American Economic Review\/}, {\em 80\/}(2), 274--279.

\bibitem[{Gintis(2006)}]{Gintis:2006}
Gintis, H. (2006).
\newblock Behavioral ethics meets natural justice.
\newblock {\em Politics, Philosophy \& Economics\/}, {\em 5\/}(1), 5--32.

\bibitem[{Joshi(1987)}]{Joshi:1987}
Joshi, N.~V. (1987).
\newblock Evolution of cooperation by reciprocation within structured demes.
\newblock {\em Journal of Genetics\/}, {\em 66\/}(1), 69--84.

\bibitem[{Samuelson(2002)}]{Samuelson:2002}
Samuelson, L. (2002).
\newblock Evolution and game theory.
\newblock {\em The Journal of Economic Perspectives\/}, {\em 16\/}(2), 47--66.

\bibitem[{Sethi \& Somanathan(2003)}]{Sethi_Somanathan:2003}
Sethi, R., \& Somanathan, E. (2003).
\newblock Understanding reciprocity.
\newblock {\em Journal of Economic Behavior \& Organization\/}, {\em 50\/}(1),
  1--27.

\bibitem[{Taylor(1987)}]{Taylor:1986}
Taylor, M. (1987).
\newblock {\em The possibility of cooperation\/}.
\newblock Cambridge University Press, revised ed.

\bibitem[{Weibull(1997)}]{Weibull:1997}
Weibull, J.~W. (1997).
\newblock {\em Evolutionary Game Theory\/}.
\newblock The MIT Press.

\bibitem[{Yao(1996)}]{YaoETAL:1996}
Yao, X. (1996).
\newblock Evolutionary stability in the n-person iterated prisoner's dilemma.
\newblock {\em Bio Systems\/}, {\em 37\/}(3), 189--197.

\bibitem[{Young \& Foster(1991)}]{Young_Foster:1991}
Young, P.~H., \& Foster, D. (1991).
\newblock Cooperation in the long-run.
\newblock {\em Games and Economic Behavior\/}, {\em 3\/}(1), 145--156.

\end{thebibliography}

\bibliographystyle{apa-good}

\end{document}